\documentclass{SCIS2026}

\usepackage{amsthm}
\usepackage{tabularx}
\newtheorem{axiom}{Axiom}
\usepackage{tikz}
\usetikzlibrary{shapes.geometric, arrows, positioning, fit, calc}
\usepackage{booktabs}  
\usepackage{makecell}
\usepackage{subcaption}

\begin{document}
\ArticleType{RESEARCH PAPER}
\Year{2025}
\Month{January}
\Vol{68}
\No{1}
\DOI{}
\ArtNo{}
\ReceiveDate{}
\ReviseDate{}
\AcceptDate{}
\OnlineDate{}
\AuthorMark{}
\AuthorCitation{}

\title{A Uniqueness Theorem for Distributed Computation under Physical Constraints}{Uniqueness Theorem for Distributed Computation}


\author[1]{Zhiyuan REN}{zyren@xidian.edu.cn}
\author[1]{Mingxuan LU}{}
\author[1]{Wenchi CHENG}{}


\address[1]{School of Telecommunication Engineering, Xidian University, No. 2 South Taibai Road, Xi'an 710071, China}

\abstract{Foundational models of computation often abstract away physical hardware limitations. However, in extreme environments like In-Network Computing (INC), these limitations become inviolable laws, creating an acute trilemma among communication efficiency, bounded memory, and robust scalability. Prevailing distributed paradigms, while powerful in their intended domains, were not designed for this stringent regime and thus face fundamental challenges. This paper demonstrates that resolving this trilemma requires a shift in perspective, one that moves from seeking engineering trade-offs to deriving solutions from logical necessity. We establish a rigorous axiomatic system that formalizes these physical constraints and prove that for the broad class of computations admitting an idempotent merge operator, there exists a unique, optimal paradigm. Any system satisfying these axioms must converge to a single normal form: Self-Describing Parallel Flows (SDPF), a purely data-centric model where stateless executors process flows that carry their own control logic. We further prove this unique paradigm is convergent, Turing-complete, and minimal. In the same way that the CAP theorem established a boundary for what is impossible in distributed state management, our work provides a constructive dual: a uniqueness theorem that reveals what is \textit{inevitable} for distributed computation flows under physical law.}


\maketitle

\section{Introduction}
Foundational models of computation, from the Turing Machine and Lambda Calculus to concurrent paradigms like the Actor Model\cite{Agha86} and Process Calculi, have provided a powerful and abstract lens through which to understand the logical limits and expressive power of computation. A common thread uniting these models is their deliberate abstraction from the physical realities of the underlying hardware; they often assume, for instance, unbounded memory, flawless communication channels, or negligible energy costs. This paper investigates a fundamental question that arises when this abstraction is removed: \textbf{What essential form must a distributed computational paradigm take when it is forced to operate under a set of inviolable physical axioms}, such as strict memory bounds and communication lower bounds, that are inherent to its execution environment? This question is no longer a purely theoretical curiosity but a pressing challenge, brought into sharp focus by the emergence of highly constrained, massively parallel computing environments such as in-network computing.

Nowhere are these physical constraints more acute than in the domain of In-Network Computing (INC). The advent of programmable network devices, such as P4 switches and Data Processing Units (DPUs), has presented an unprecedented opportunity for line-rate data processing, but this new paradigm operates under conditions that fundamentally distinguish it from traditional distributed computing. Specifically, a devices's fast memory is severely limited and physically independent of the total problem size ($W_p=O(1)$), while the network environment it inhabits is inherently unreliable. These physical laws impose a stringent, unavoidable trilemma: the concurrent optimization of communication efficiency, bounded memory, and robust scalability. Prevailing paradigms like the Bulk Synchronous Parallel (BSP) model \cite{Valiant90} or stateful message-passing, designed for more resource-rich environments, falter under these conditions, which necessitates a return to first principles to identify the viable computational models that can operate within such severe limitations.

This papers central thesis is that in such a physically-constrained environment, the key to resolving the trilemma lies not in a novel protocol, but in an intrinsic algebraic property of the computational problem itself: the existence of an idempotent, commutative, and associative merge operator. We assert that for the broad class of problems satisfying this property, the trilemmas constraints can be simultaneously met. Crucially, within the INC context, leveraging this property is not merely an optimization but a prerequisite for feasibility. To circumvent costly coordination and state recovery on memory-starved devices, the computational model must undergo a fundamental shift from being process-centric to data-centric. That is, the logic for fault tolerance and flow control must migrate from complex inter-process protocols to the structure of the data flows themselves.

Therefore, this paper transcends the proposal of another better system for INC, aiming instead to establish a fundamental law for this physically-constrained class of computations. We develop a formal model and an axiomatic system that abstracts the physical constraints of INC (communication lower bounds, heavy-tailed network perturbations, and constant memory per node)into inviolable mathematical rules. Within this system, we prove that any paradigm satisfying these axioms must converge to the unique, data-centric normal form of Self-Describing Parallel Flows (SDPF). Our contribution is a complete set of formal proofs for this paradigm, including its uniqueness (via constructive reduction), strong eventual consistency, Turing completeness, and the minimality of its design. The remainder of this paper is dedicated to the rigorous mathematical construction and proof of this central claim.

\section{Theoretical Context and Related Models of Computation}

To clearly position our theoretical contribution, we situate the SDPF theorem within the landscape of distributed computing theory, viewed through the specific lens of In-Network Computings physical constraints. We compare our work with seminal results in several key areas to elucidate its unique boundaries and core contributions.

\subsection{In-Network Computing Models and Systems}

The field of In-Network Computing has produced several powerful programming models and systems, such as  \cite{Bosshart2014P4} and various frameworks for DPUs. These systems provide essential abstractions for directly programming the network data plane. However, their primary focus is often on expressing packet processing logic and achieving line-rate performance. Issues like robust scalability, fault tolerance, and memory management under heavy traffic are typically addressed through implementation-specific, ad-hoc techniques rather than a guiding theoretical framework.

These systems, therefore, serve as a compelling manifestation of the problem we address: they are powerful engineering artifacts whose necessary design patterns and trade-offs have been discovered through practice, but without a formal theory to explain their inevitability. This paper provides that missing theoretical foundation, deriving from first principles the unique paradigm (SDPF) that must emerge under the physical laws governing these systems.\cite{Kianpisheh2023InNetworkSurvey, Hauser2023P4Survey,Bosshart2014P4}

\subsection{Communication and I/O Lower Bounds}

In the INC context, the on-chip memory of a network device and the hosts main memory form a classic two-level storage hierarchy, making communication lower bounds acutely relevant. Classic works, using models such as the red-blue pebble game and communication-avoiding algorithms, have characterized the unavoidable communication costs in such systems (e.g., \cite{HK81, ITT04, BDHS11}).

These foundational results establish the unavoidable cost of data movement. Our work does not seek to challenge these bounds; rather, we elevate this physical law to the status of a formal axiom (Axiom A1) and proceed to ask a new, composite question: \textbf{Given that this communication cost is an inviolable lower bound, what operational form must a paradigm take to also satisfy the constraints of an unreliable environment (A2) and strictly bounded memory (A3)?} By integrating these physical and operational constraints, we move from analyzing a single dimension of cost to deriving a holistic, necessary paradigm.

\subsection{Logical Convergence: CALM and Semilattices}

The challenge of achieving consistency without costly coordination has been masterfully addressed at the logical level. The CALM (Consistency as Logical Monotonicity) theorem provides the fundamental principle: monotonic programs are guaranteed to be eventually consistent without coordination \cite{Hellerstein12}. This principle finds its algebraic manifestation in semilattice structures, as elegantly formalized by Conflict-free Replicated Data Types (CRDTs). CRDTs leverage commutative, associative, and idempotent merge operators to ensure convergence regardless of network delays, reordering, or duplication \cite{Shapiro11}.

This powerful body of work forms the theoretical underpinning for our Axiom A4. It definitively answers \textit{what} algebraic properties are required for logical convergence. However, the scope of this work is confined to the logical realm; it does not prescribe an operational paradigm that can concurrently satisfy the stringent physical demands of Axiom A1 (Communication Lower Bound) and Axiom A3 (Memory Upper Bound). Our research, therefore, begins where this foundational work on logical consistency leaves off. We ask: \textbf{When the logical necessity of semilattice convergence (A4) is combined with the physical necessities of communication and memory bounds (A1, A3), what is the unique operational form that a system must take?}

\subsection{Operational Models and their Physical Infeasibility}

While the models discussed above address logical convergence, the dominant operational paradigms in distributed computing were architected with different, more forgiving physical assumptions. Paradigms like Valiants Bulk Synchronous Parallel (BSP) model, and its descendants in batch \cite{DeanGhemawat2004MapReduce, Zaharia2012RDD}  and streaming (Flink/Beam), rely on coordination mechanisms such as global barriers, checkpoints, or ordered watermarks \cite{Akidau2015DataflowModel,Lamport1978Clocks}. As formally proven in our Proposition 5.3, any form of global synchronization is fundamentally incompatible with Axiom A2 (Perturbation and Failures), as heavy-tailed latency inevitably destroys linear scalability. These paradigms are therefore physically infeasible in the INC domain.

Similarly, classic logical models for concurrency, while foundational, abstract away the very constraints we consider paramount. Kahn Process Networks (KPN) guarantee determinism through the assumption of unbounded FIFO channels, a direct violation of our Axiom A3 (Memory Upper Bound) \cite{Kahn74}. The Actor model, a powerful abstraction for composable concurrency, focuses on behavioral guarantees and does not inherently address the physical lower bounds of communication (Axiom A1) or memory (Axiom A3) \cite{Agha86}. Our work does not refute the logical elegance of these models; instead, by elevating physical laws to first-class axioms, we prove the necessary operational form a concurrent model must take when it can no longer ignore these harsh realities.

More advanced dataflow systems, such as Naiad \cite{Murray2013Naiad}  or Differential Dataflow \cite{McSherry2013DifferentialDataflow}, have replaced coarse-grained barriers with more sophisticated progress tracking mechanisms like timestamps or epochs. While this represents a significant improvement, these mechanisms still constitute a form of logical coordination that is fundamentally vulnerable within our axiomatic regime. A single straggler suffering from heavy-tailed latency (Axiom A2) can halt the progress of a timestamp or epoch, preventing the system from reclaiming resources or committing results. This reintroduces a subtle form of the convoy effect and can lead to memory pressure that violates the strict bounded memory guarantee (Axiom A3), as the system is forced to buffer ever-increasing amounts of out-of-order data. Thus, even these state-of-the-art models are incompatible with the strict physical realities we study.

\subsection{A Constructive Duality to Impossibility Theorems}

Classic impossibility results, most notably the FLP and CAP theorems, define the fundamental boundaries of what is achievable in asynchronous distributed systems facing failures \cite{FLP85, Gilbert02}. Our work offers a constructive dual: rather than proving what is impossible to achieve (e.g., consensus), we prove what is inevitable for a different class of problems. We start by accepting a set of axioms, including one that explicitly circumvents the need for consensus (Axiom A4), and derive the unique computational form that is not just possible, but necessary.

This perspective provides a formal underpinning for navigating challenges like the tail at scale phenomenon \cite{Dean-Barroso13}. In the physically-constrained environments we study, a single straggler, a direct consequence of the heavy-tailed perturbations formalized in Axiom A2, is not merely a performance issue but a threat to system survival due to bounded memory (Axiom A3). While impossibility theorems erect walls, our uniqueness theorem illuminates the single, narrow path that remains. By proving from first principles that barrier-less, data-centric computation is a condition for survival, we provide a constructive guide for building robust systems in this challenging domain.

\subsection{Summary of Theoretical Positioning} 
The preceding survey situates our work within the landscape of distributed computing theory. While foundational models provide powerful abstractions for parallelism and consistency, they were often designed under a different set of physical assumptions. The emergence of highly constrained environments like In-Network Computing presents a unique synthesis of challenges—communication, memory, and robustness—that motivates a return to first principles. 

This inquiry leads us to ask: \textbf{What set of properties must a computational paradigm possess to satisfy all these physical and logical constraints simultaneously?} To crystallize this challenge, Table~\ref{tab:model-comparison} summarizes how existing foundational models measure up against these core requirements, thereby revealing a critical gap that this paper aims to address.

\begin{table*}[!t]
\footnotesize
\caption{Comparison of Computational Models Against Core Axiomatic Requirements}
\label{tab:model-comparison}
\centering
\tabcolsep 20pt 
\begin{tabular*}{\textwidth}{@{\extracolsep{\fill}}lcccc}
\toprule
 {\bf Model / Requirement} & 
 {\begin{tabular}[c]{@{}c@{}} \bf Handles Comm. \\ \bf Bound (A1) \end{tabular}} & 
 {\begin{tabular}[c]{@{}c@{}} \bf Handles \\ \bf Perturbations (A2) \end{tabular}} & 
 {\begin{tabular}[c]{@{}c@{}} \bf Handles Memory \\ \bf Bound (A3) \end{tabular}} & 
 {\begin{tabular}[c]{@{}c@{}} \bf Is Consensus-Free \\ \bf (A4) \end{tabular}} \\
\hline
 {\bf This Work's Goal} & {\bf Yes} & {\bf Yes} & {\bf Yes} & {\bf Yes} \\
 \hline
 BSP Model~\cite{Valiant90}      & No  & No  & No  & No \\
 KPN~\cite{Kahn74}       & --  & Yes & No  & Yes \\
 Actor Model~\cite{Agha86}    & --  & Yes & --  & Yes \\
 Dataflow~\cite{Murray2013Naiad}       & No  & No  & No  & No \\
 CRDTs~\cite{Shapiro11}          & --  & Yes & --  & Yes \\
\bottomrule
\end{tabular*}
\vspace{2mm}
\par 
\centering
\parbox{\textwidth}{\small {\bf Note:} This table illustrates the core challenge. We posit that a viable paradigm in this regime must satisfy all four requirements simultaneously. Existing models, while powerful in their respective domains, are shown to be incompatible with at least one of these foundational constraints. The remainder of this paper is dedicated to formally deriving the unique paradigm that meets this objective.}
\end{table*}

\section{Formal Model and Axioms}

Having established the theoretical context and the shortcomings of existing models, we now construct the formal foundation upon which our uniqueness theorem is built. This chapter is not merely a description of a model; it is an act of formalizing the physical laws of a constrained computational regime. We abstract the realities of In-Network Computing into a precise mathematical framework, defining the execution model, the metrics for evaluating its performance against the trilemma, and, most critically, the core axioms that any viable paradigm within this regime must obey.

\subsection{Execution Model and Cost Metrics}

\textbf{Execution Model}: To enable rigorous derivation, we define a general model for distributed computation that captures the essence of the INC environment. We consider a set of processes (nodes) $P=\{1,\dots,|P|\}$, representing programmable network devices, that communicate through asynchronous, reliability-limited message channels. A complete execution trace is a (potentially infinite) interleaved sequence of message sends, receives, and local computation steps across all nodes. Each node possesses a private fast memory (e.g., on-chip SRAM) with an upper capacity of $W_p$ (its specific constraints are detailed in Axiom A3). The system also has a logically global, consistent slow storage layer (e.g., host main memory) for the initial input $I$ and final output $O$.

\textbf{Data Model}: The atomic unit of computation is a tile or an atomic increment $\delta$ generated during computation. To achieve coordination-free asynchronous processing, each such data unit must carry a minimal set of self-describing metadata:
$$m=(\text{id}, \text{target}, \text{op}, \text{pc/next}, \text{rid})$$
Here, id is the unique identifier of the tile; target indicates the destination to which the increment should be merged; op and pc/next define the stateless computation to be performed and the subsequent control flow; and rid (replay id) is a globally unique identifier for the same semantic contribution, used to achieve idempotency.

\textbf{Cost Metrics}: Based on this model, we define three key cost metrics corresponding to the three dimensions of the trilemma, which serve as the quantitative basis for our analysis:
\begin{enumerate}
    \item \textbf{I/O Amplification Factor ($R_A$)}: $R_A := \frac{Q}{|I|+|O|}$, where $Q$ is the total number of bytes communicated across nodes during execution. This metric measures communication efficiency by representing the factor by which the actual communication ($Q$) is amplified over the necessary I/O for the initial input ($I$) and final output ($O$).
    \item \textbf{Peak Memory Footprint ($W_{\max}$)}: $W_{\max} := \max_{p\in P}\ \max_{t}\ \text{(fast memory usage of node } p \text{ at time } t)$. This metric measures memory boundedness.
    \item \textbf{Scalability ($S(P)$)}: Defined as the functional relationship between the systems steady-state throughput and the number of nodes $|P|$ for a given workload. Ideal linear scalability implies $S(P) = \Theta(P)$.
\end{enumerate}

\subsection{Axiomatic System}

We now propose four axioms that formalize the computational regime. These axioms are not arbitrary; they are derived from and classified by their origin. They form the logical cornerstone of our uniqueness proof. Axioms A1 and A3 formalize the inviolable \textit{physical resource constraints} (communication and memory). Axiom A2 formalizes the \textit{environmental constraints} of a realistic distributed system (perturbations and failures). Finally, Axiom A4 defines the \textit{algebraic properties} of the class of computational problems we address. Our central thesis rests on proving the unique paradigm that emerges from the intersection of all four constraints.

\begin{axiom}[Communication Lower Bound]
\label{ax:comm}
In a two-level memory model, for any computational problem requiring a non-trivial amount of data reuse, I/O complexity theory establishes a fundamental trade-off between local memory size and the necessary cross-node communication volume $Q$. Any correct algorithm is subject to a theoretical lower bound on this volume.
\begin{quote}
\textbf{Formalization}: The total communication volume is bounded by $Q = \Omega(f(n, W_p))$, where $f$ is a function of the problem size $n$ and the local fast memory size $W_p$. For many important classes of algorithms, such as those with a dependency graph structure akin to the red-blue pebble game, this bound is polynomial in $n$ and inversely related to some power of $W_p$. For communication-optimal algorithms, $Q$ must be at least $\Omega(|I|+|O|)$.
\end{quote}
\end{axiom}

\begin{axiom}[Perturbation and Failures]
\label{ax:perturb}
The environment of distributed computation is inherently unreliable. Task latencies follow a heavy-tailed distribution, and nodes may fail with a nonzero probability. The network permits message reordering, loss, and duplication. We assume a standard fair-lossy channel model which guarantees that if a message is re-sent infinitely often, it will eventually be delivered. This eventual delivery (fairness) assumption is a critical boundary condition for our convergence proofs.
\begin{quote}
\textbf{Formalization}: (i) The service time $\tau$ of each atomic task is an i.i.d. random variable whose distribution is heavy-tailed (i.e., its tail decays polynomially or slower, $\Pr[\tau > t] \sim t^{-\alpha}$ for some $\alpha > 0$). (ii) Each node fails with a constant probability $p_f > 0$.
\end{quote}
\end{axiom}

\begin{axiom}[Memory Upper Bound]
\label{ax:memory}
The local fast memory capacity of a single node is independent of the global size of the problem.
\begin{quote}
\textbf{Formalization}: $W_p=O(1)$, where the asymptotic bound is relative to the global problem size $n$ (it is only on the order of the tile or increment of granularity).
\end{quote}
\end{axiom}

\begin{axiom}[Idempotent and Commutative Merge]
\label{ax:merge}
This axiom defines the algebraic property of the problem domain for which our uniqueness theorem holds. We focus on the broad and significant class of computations whose partial results can be correctly combined via a merge operator that is commutative, associative, and idempotent. This class includes many large-scale distributed tasks, such as telemetry aggregation (e.g., finding a maximum value), statistical estimations (e.g., building a histogram), and certain coordination-free machine learning updates. It explicitly excludes tasks requiring strict, linearizable state transitions, such as transactional banking systems.
\begin{quote}
\textbf{Formalization}: For each target, its state space forms an algebraic structure $(\mathcal{S},\oplus)$, where the merge operator $\oplus:\mathcal{S}\times\mathcal{S}\to\mathcal{S}$ is commutative, associative, and idempotent for the same rid. The global state space of the entire system is a Cartesian product semilattice $\prod_{t}(\mathcal{S}_t,\oplus)$.
\end{quote}
\end{axiom}

\vspace{1em}
\textit{Note: A1–A4 will be used as axioms. This paper does not aim to re-prove these classic communication/I/O lower bound theories or semilattice convergence theories, but rather to use them as the starting point for our deductive reasoning.}

\section{The SDPF Paradigm}

\begin{figure}[!t]
\centering
\begin{minipage}[c]{0.48\textwidth}
\centering
\includegraphics[width=\textwidth]{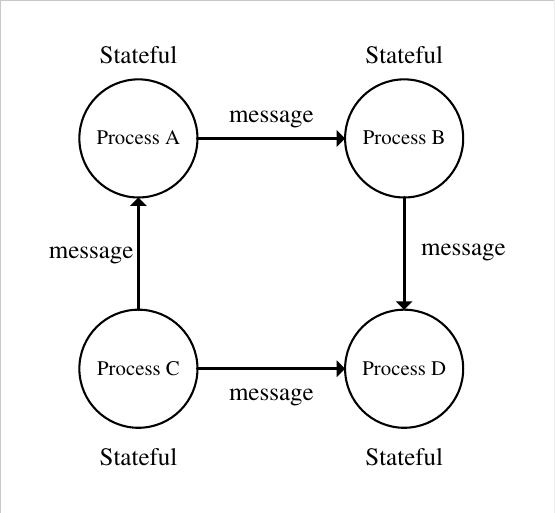}
\centerline{(a)}
\end{minipage}
\hspace{0.02\textwidth}
\begin{minipage}[c]{0.48\textwidth}
\centering
\includegraphics[width=\textwidth]{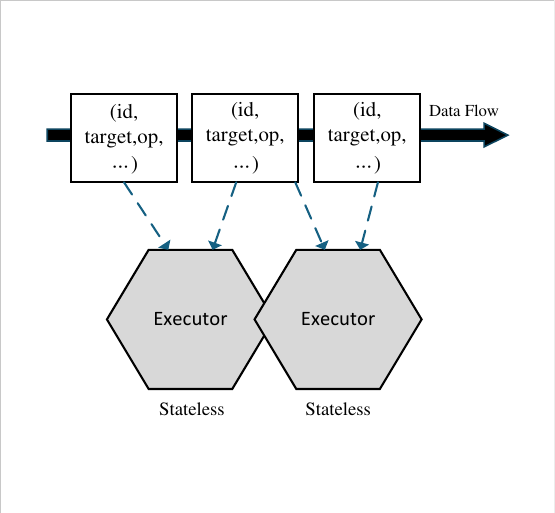}
\centerline{(b)}
\end{minipage}
\caption{A conceptual comparison between traditional process-centric models and the data-centric SDPF paradigm. (a) Process-Centric Model; (b) The SDPF Data-Centric Model. In SDPF, intelligence and control logic are embedded in the data flow, while executors are stateless and anonymous.}
\label{fig:paradigm-comparison}
\end{figure}

The axiomatic system established in Chapter 3 defines a rigorous and challenging computational regime. Any paradigm that operates within it is not a matter of arbitrary design, but of logical necessity. In this chapter, we derive the unique structure that emerges as the inevitable consequence of these axioms: the Self-Describing Parallel Flows (SDPF) paradigm.

Its core is a fundamental and forced shift from being process-centric to data-centric. We will demonstrate that in order to satisfy the axiomatic constraints, the responsibility for state management and control flow must be externalized from the processing nodes, rendering them stateless and anonymous execution engines. While the computational logic (the op) remains at the node, the intelligence that guides the stateful, end-to-end progression of the computation is embedded into the data flows themselves. This role reversal is not a design choice but the only viable mechanism for survival. The five properties (S1-S5) defined below are, therefore, not a menu of features, but the necessary and interlocking components of this unique normal form.

\begin{definition}[SDPF Normal Form]
An execution is said to be SDPF if and only if it simultaneously satisfies the following five properties (S1-S5):
\end{definition}

\subsection{Single-Read Reuse (S1)}
This property directly addresses the requirement of Axiom \ref{ax:comm} (Communication Lower Bound). To minimize communication overhead, the process of bringing data into the system from the expensive slow storage layer must be free of redundancy. Subsequent data sharing is accomplished through cheaper in-network copying or multicasting.

\begin{quote}
    (S1) Single-Read Reuse: Each input tile is introduced into the system from slow memory at most once. Subsequent acquisition by multiple consumers is achieved via network-level replication/reuse.
\end{quote}

\subsection{Stateless Micro-tasks (S2)}
This property is a crucial step in dealing with Axiom \ref{ax:perturb} (Perturbation and Failures). By decomposing computation into pure-functional micro-tasks, the system eliminates the difficult problem of state recovery after node failures. This allows any task to be safely replayed on any node, laying the foundation for high fault tolerance.

\begin{quote}
    (S2) Stateless Micro-tasks: Each computation is a pure function application on its carried data, without reading or modifying any cross-task, mutable shared state.
\end{quote}

\subsection{Idempotent Merge (S3)}
This property works in close concert with S2 to form the core of robustness. It ensures that in the reordered and duplicated network environment described by Axiom \ref{ax:perturb}, the correctness of the result is unaffected. This is achieved by relying on the algebraic structure defined in Axiom \ref{ax:merge}.

\begin{quote}
    (S3) Idempotent Merge: Each increment $\delta$ produced by a computation carries a unique replay identifier rid. The target state is updated via a merge operator $\oplus$ that is commutative and associative; merging duplicate increments with the same rid does not change the target state.
\end{quote}

\subsection{Barrier-less Asynchronous Scheduling (S4)}
This property is designed to combat the heavy-tailed latency defined in Axiom \ref{ax:perturb}. By eliminating global synchronization points, the system avoids the convoy effect, where overall performance is dictated by the slowest task. This makes it possible for the systems scalability $S(P)$ to approach ideal linear growth.

\begin{quote}
    (S4) Barrier-less Asynchronous Scheduling: The execution process contains no global barriers that require all or a large number of nodes to be ready before proceeding. Any ready task whose dependencies have been met can be scheduled for immediate execution.
\end{quote}

\subsection{Sliding Window (S5)}
This property is a direct satisfaction of Axiom \ref{ax:memory} (Memory Upper Bound). By limiting the amount of active data that a node needs to buffer at any given time, it ensures that a single nodes memory footprint is decoupled from the global problem size $n$, thereby enabling true horizontal scalability.

\begin{quote}
    (S5) Sliding Window: Each node only buffers a finite window of active data that is either awaiting future computation (e.g., waiting for a pair) or being accumulated. This ensures the peak memory footprint $W_{\max}$ is a constant independent of $n$, i.e., $W_{\max}=O(1)$.
\end{quote}

These five properties collectively define the complete form of the SDPF paradigm. They are not isolated but are interdependent, working together to satisfy the stringent axioms we established in Chapter 3. In the following chapters, we will demonstrate through a series of rigorous mathematical proofs why these five properties are not only sufficient but also necessary.

\section{Proof of Necessity and Uniqueness}

The preceding chapters have established a formal axiomatic system and precisely defined the SDPF paradigm. We now arrive at the central argument of this paper: to prove that SDPF is not merely a well designed option, but the unique and necessary outcome of our axiomatic constraints. Our proof strategy is a constructive one. We will demonstrate through a series of necessity propositions that any deviation from the core properties of SDPF inevitably leads to a violation of one or more axioms. This method reveals that the SDPF normal form is not so much designed as it is discovered, an inescapable consequence of the physical laws governing the computational regime. To clarify the overall proof structure, we first provide the following roadmap.

\begin{table}[!t]
\footnotesize
\caption{Proof Roadmap for SDPF Uniqueness}
\label{tab:roadmap}
\tabcolsep 6pt 
\begin{tabular*}{\textwidth}{@{\extracolsep{\fill}}lllll}
\toprule
\textbf{Trilemma Goal} & \textbf{Required SDPF Property} & \textbf{If Absent...} & \textbf{...Axiom Violated} & \textbf{Proposition} \\ \hline
Communication Efficiency & Single-Read Reuse & I/O is amplified & A1 (Comm. Lower Bound) & Prop. 5.1 \\
Robustness & Stateless/Idempotent & Fault tolerance is costly & A1 / A2 & Prop. 5.2 \\
Scalability & Barrier-less Async. & Tail latency is amplified & A2 (Perturbation Model) & Prop. 5.3 \\
Bounded Memory & Sliding Window & Memory is coupled with scale & A3 (Memory Upper Bound) & Prop. 5.4 \\
\bottomrule
\end{tabular*}
\end{table}

Next, we will prove these four propositions one by one.

\subsection{Necessary Condition for Communication Efficiency}

\begin{proposition}[Necessity of Single-Read Reuse]
\label{prop:single-read}
Let the set of tiles be $\{\text{tile}_k\}$ with consumption counts $U_k$. If an algorithm allows some consumers to independently fetch the same tile from the source (without in-network reuse), then
$$R_A \;\ge\; 1 + \frac{\sum_k (U_k-1)\cdot|\text{tile}_k|}{|I|+|O|}.$$
In particular, when some $U_k=\Omega(n)$, then $R_A=\Omega(n)$.
\end{proposition}

\begin{proof}[Proof Sketch]
This proof aims to quantify the communication overhead incurred by violating the single-read reuse principle. Our strategy is to express the lower bound of the total communication volume $Q$ as the sum of two components: the unavoidable baseline I/O, which is $|I|+|O|$, and the redundant I/O generated by each tile $\text{tile}_k$ being repeatedly fetched from the source $U_k-1$ times. By substituting this lower bound of $Q$ into the definition of the I/O amplification factor $R_A$, we can directly derive that $R_A$ must be strictly greater than 1. The excess part is proportional to the normalized redundant communication volume. To satisfy the communication optimality Axiom \ref{ax:comm} (i.e., $R_A=O(1)$), this redundant term must be eliminated, thus proving the necessity of single-read reuse (S1). The detailed formal proof is provided in  \ref{sec:appendix-prop51}.
\end{proof}

\subsection{Necessary Conditions for Robustness}

\begin{proposition}[Necessity of Statelessness and Idempotency]
\label{prop:stateless-idempotent}
Under Axiom \ref{ax:perturb}, if a micro-task reads/modifies a state shared across tasks, or if the target merge does not satisfy the properties of Axiom \ref{ax:merge}, then any algorithm that guarantees correctness must satisfy one of the following:
\begin{enumerate}
    \item[(C1)] It introduces per-step/per-batch persistent writes (checkpoints or write-ahead logs), making it impossible for $R_A$ to remain $O(1)$.
    \item[(C2)] It introduces global sequential coordination/barriers, thereby falling into the scalability lower bound of Proposition \ref{prop:barrier-less}.
\end{enumerate}
\end{proposition}

\begin{proof}[Proof Sketch]
This proof aims to reveal that, in the environment of perturbations and failures described by Axiom \ref{ax:perturb}, any paradigm relying on mutable state or non-idempotent operations inevitably falls into a dilemma. To ensure correctness, such systems must choose between two costly fault-tolerance strategies. The first is to remember through persistence (e.g., checkpoints), which generates enormous communication overhead, violating Axiom \ref{ax:comm}. The second is to enforce order (e.g., total-order broadcast), which is equivalent to a global barrier and, according to Proposition \ref{prop:barrier-less}, destroys linear scalability. Since both paths lead to the violation of core axioms, the paradigm must adopt stateless micro-tasks (S2) and idempotent merge (S3). The detailed formal proof is provided in \ref{sec:appendix-prop52}.
\end{proof}

\begin{proposition}[Necessity of Barrier-less Asynchronous Scheduling]
\label{prop:barrier-less}
Let a barrier-synchronized parallel round consist of independent latencies $\tau_1,\dots,\tau_{|P|}$. If the tail of $\tau$s distribution is heavy (per Axiom \ref{ax:perturb}), then the expected round completion time, $T_{\mathrm{round}}=\max_{i\le |P|}\tau_i$, is an increasing function of $|P|$. Consequently, throughput $S(P)$ is sub-linear, making $S(P)=\Theta(P)$ unattainable.
\end{proposition}

\begin{proof}[Proof Sketch]
This proof reveals from first principles why synchronization barriers combined with heavy-tailed latency (Axiom \ref{ax:perturb}) fundamentally destroy linear scalability. The convoy effect of barriers means the round time is determined by the slowest task. Based on principles from extreme value theory, for any heavy-tailed distribution, the expectation of the maximum of $|P|$ i.i.d. variables is a monotonically increasing function of $|P|$. This is in stark contrast to distributions with exponentially decaying tails (e.g., Gaussian), where the expected maximum grows much more slowly, typically logarithmically. Since throughput $S(P) \propto |P|/E[T_{\mathrm{round}}]$, an increasing round time necessarily leads to sub-linear scalability ($S(P)=o(P)$). Therefore, barrier-less scheduling (S4) is a necessary condition. The detailed formal proof is provided in \ref{sec:appendix-prop53}.
\end{proof}

\subsection{Necessary Condition for Bounded Memory}

\begin{proposition}[Necessity of Sliding Window]
\label{prop:sliding-window}
If an algorithm hoards unused inputs within a node to reduce communication frequency, then there exists a $\beta>0$ such that $W_{\max} = \Omega(n^{\beta})$, thus violating the memory upper bound of Axiom \ref{ax:memory}.
\end{proposition}

\begin{proof}[Proof Sketch]
This proof originates from I/O complexity theory, which reveals a fundamental trade-off between memory size and communication volume. To approach the communication lower bound (Axiom \ref{ax:comm}), an algorithm must maximize data reuse, which incentivizes hoarding data in memory. However, I/O complexity theory proves that for many problems, the size of the live dataset required to achieve optimal communication grows polynomially with the problem size $n$. This leads to $W_{\max} = \Omega(n^{\beta})$, which directly contradicts the bounded memory Axiom \ref{ax:memory} ($W_{\max}=O(1)$). The only way to resolve this conflict is to adopt a use and discard pipelined model, which is the essence of the sliding window (S5). The detailed formal proof is provided in  \ref{sec:appendix-prop54}.
\end{proof}

\subsection{The Uniqueness Main Theorem}

\begin{theorem}[Uniqueness up to metric-equivalence]
\label{thm:uniqueness}
Under Axioms A1--A4, every correct paradigm is transformable by T1--T4 into an SDPF instance that is unique as a normal form modulo metric-equivalence (Def. 3.2), i.e., their communication, peak memory, and scalability differ by at most constant factors independent of $n$ and $|P|$.
\end{theorem}

\begin{proof}[Proof Sketch]
This proof employs a constructive reduction method. We start with an arbitrary optimal paradigm $\mathcal{E}$ and apply a series of transformations (T1-T4). Propositions \ref{prop:single-read} through \ref{prop:sliding-window} provide the theoretical basis, showing that each correction is necessary to satisfy a core axiom. Crucially, each transformation is shown to be performance-preserving. The overheads introduced, such as the management of replay identifiers (rid) in T2 or the buffering required by T4, are shown to be constant or bounded, and most importantly, independent of the system scale $|P|$ and network topology. The resulting paradigm $\mathcal{E}$ is, by definition, an SDPF normal form and is performance-equivalent to $\mathcal{E}$. This proves that SDPF is the unique normal form for all optimal paradigms. The full proof with formal preservation lemmas for each transformation is provided in  \ref{sec:appendix-thm55}.
\end{proof}

\begin{figure}[h!]
  \centering
  \includegraphics[width=0.6\columnwidth]{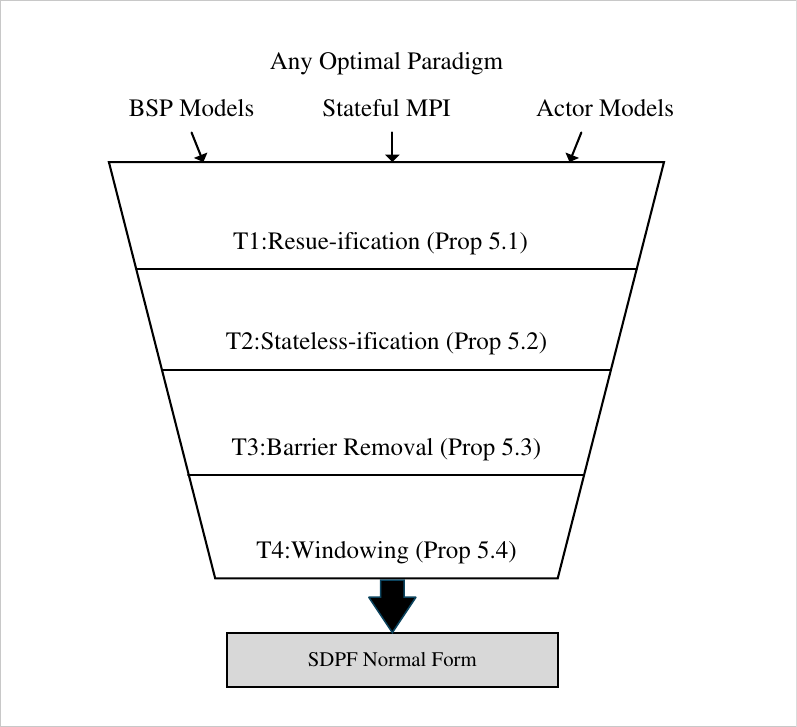}
  \caption{The constructive reduction proof of Theorem 5.5. Any paradigm claiming optimality under the axioms is progressively transformed by eliminating inefficiencies (T1-T4), inevitably converging to the unique SDPF normal form.}
  \label{fig:reduction-funnel}
\end{figure}

\begin{corollary}[Impossibility Dual]
\label{cor:impossibility}
For any execution paradigm $\mathcal{E}$, if $\mathcal{E}\not\equiv$ SDPF, then at least one of the trilemma goals is violated: either $R_A=\omega(1)$, or $W_{\max}=\omega(1)$, or $S(P)=o(P)$.
\end{corollary}

\section{Convergence, Completeness, and Minimality}

Chapter 5 established the uniqueness and necessity of the SDPF normal form as a consequence of physical axioms. This chapter addresses the subsequent fundamental question: Is this necessary paradigm also a sufficient and robust model of computation? To answer this, we formally prove that SDPF satisfies three canonical properties required of any complete theoretical paradigm. First, we prove its guaranteed \textit{convergence}, ensuring correctness in a chaotic environment. Second, we establish its \textit{Turing completeness}, demonstrating its universal expressive power. Finally, we prove the \textit{minimality} of its design, confirming its theoretical elegance and efficiency. Together, these properties solidify SDPFs status as a complete and foundational model for computation under physical constraints.

\subsection{Convergence}

\begin{definition}[Semilattice Semantics]
For each target, its state space forms a semilattice $(\mathcal{S},\oplus)$ where the merge operator $\oplus$ is commutative, associative, and idempotent (per Axiom \ref{ax:merge}). The global state is the Cartesian product semilattice $\mathcal{L}=\prod_t(\mathcal{S}_t,\oplus)$.
\end{definition}

\begin{theorem}[Strong Eventual Consistency]
\label{thm:consistency}
Under SDPF and Axiom \ref{ax:perturb} fairness, any state sequence $(A_0,A_1,\dots)\subseteq \mathcal{L}$ produced by an execution trace is monotonically non-decreasing and converges to a unique limit, which is independent of message arrival order and repetition.
\end{theorem}

\begin{proof}[Proof Sketch]
The theoretical guarantee for this deterministic outcome in a chaotic environment comes from the semilattice algebraic structure of the systems state space (Axiom \ref{ax:merge}). The core idea is that every merge operation ($\oplus$) causes the system state to monotonically move upward in the semilattices partial order. The commutativity and associativity of $\oplus$ guarantee independence from order, while idempotency guarantees independence from repetition. Since all execution paths advance monotonically towards a shared upper bound, they must converge to the same unique limit. The detailed formal proof is provided in  \ref{sec:appendix-thm62}.
\end{proof}

\subsection{Completeness}

To prove the universal computational capability of SDPF (i.e., Turing completeness), we adopt a standard strategy from theoretical computer science: proving that SDPF can fully simulate a primitive, Turing-complete model. We choose the SK combinator calculus as our simulation target.

Before delving into the technical details, it is crucial to clarify the core challenge addressed here. A logical equivalence like $Sxyz \to xz(yz)$ is trivial to implement in a centralized, single-machine environment. The complexity of our proof, however, arises from a fundamentally different question: How can a distributed system, composed of stateless nodes communicating over an unreliable network (as defined by Axiom A2), be guaranteed to atomically and deterministically achieve this logical state transition? Therefore, the focus of the following proof is not to reiterate the logical equivalence itself, but to demonstrate that SDPFs distributed mechanisms are capable of realizing this universal computational process. This establishes that the true challenge—and contribution—lies in ensuring this equivalence holds within a distributed physical reality.

\begin{definition}[Translation $\Phi$]
We define a translation function $\Phi$ from the SK combinator calculus to a multiset of SDPF tiles. SK nodes are mapped to tiles. Branch detects reducible expressions (redexes). Copy and Combine perform the rewrite by injecting new tiles and marking old ones with a tombstone, using a shared rid for atomicity.
\end{definition}

\begin{lemma}[One-Step Reduction Preservation]
\label{lemma:one-step}
If $M\to N$ is a basic SK reduction, then there exists a finite sequence of SDPF transformations such that $\Phi(M)\ \Longrightarrow^*\ \Phi(N)$.
\end{lemma}

\begin{proof}[Proof Sketch]
This lemma is the cornerstone for proving Turing completeness. It shows that SDPFs core primitives (Branch, Copy, Combine) are sufficient to simulate an atomic find-and-replace step of the SK calculus. The Branch primitive finds a reducible pattern, while Combine and Copy replace it by atomically injecting new tiles and retiring old ones. The robustness of this simulation against network chaos is guaranteed by the statelessness and idempotency of SDPF operations. The detailed construction is provided in \ref{sec:appendix-lem64}.
\end{proof}

\begin{figure}[h!]
  \centering
  \includegraphics[width=0.9\columnwidth]{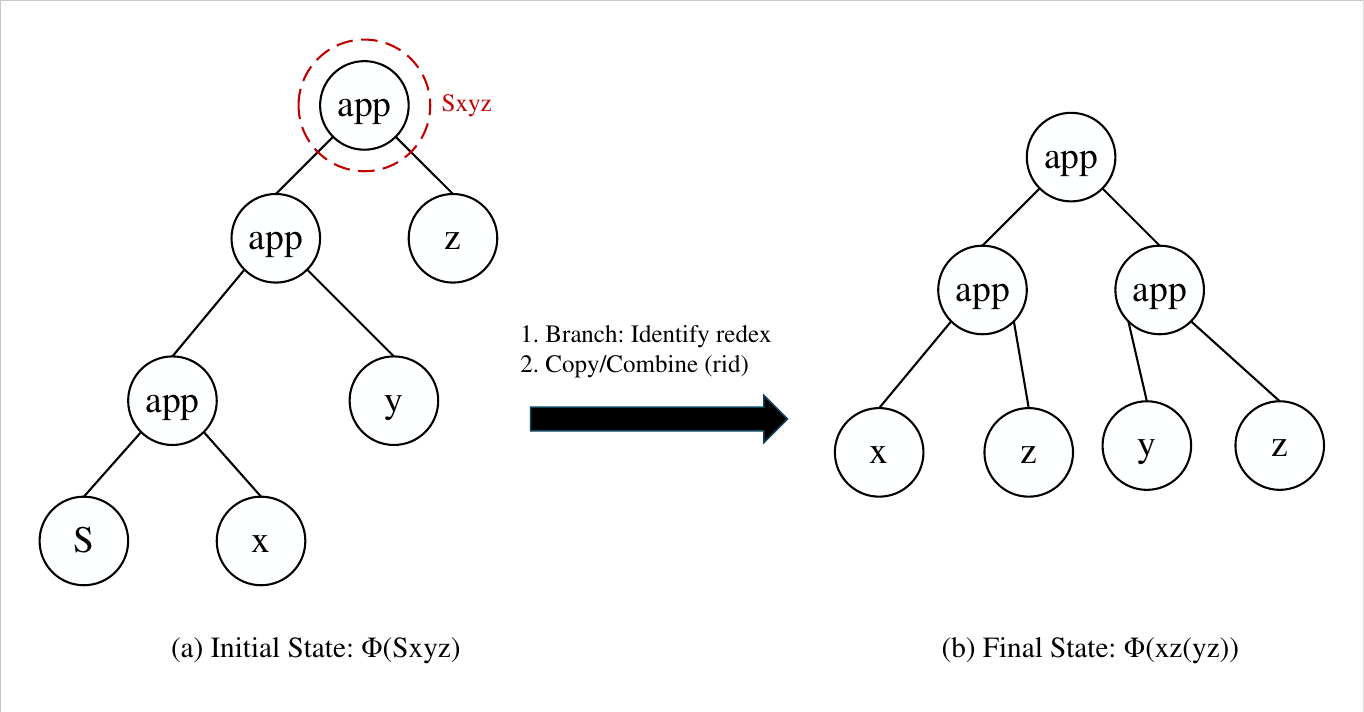}
  \caption{Simulating an SK-combinator reduction ($Sxyz \to xz(yz)$) using SDPF primitives. The transformation is guaranteed to be atomic and idempotent by the Combine operation and its unique rid.}
  \label{fig:sk-simulation}
\end{figure}

\begin{theorem}[Turing Completeness]
\label{thm:turing-complete}
SDPF is Turing-complete. Any fair SDPF execution can simulate any SK reduction sequence.
\end{theorem}

\begin{proof}[Proof Sketch]
This proof elevates the one-step simulation of Lemma \ref{lemma:one-step} to the simulation of any full computation. Lemma \ref{lemma:one-step} guarantees the ability to execute each step correctly, while the fairness assumption in Axiom \ref{ax:perturb} guarantees that any ready step will eventually be executed. This combination ensures that SDPF can simulate any SK reduction sequence. Since SK calculus is Turing-complete, SDPF is as well. The detailed formal proof is provided in  \ref{sec:appendix-thm65}.
\end{proof}

\subsection{Minimality}

\begin{theorem}[Necessity of Minimal Metadata Set]
\label{thm:minimality}
The absence of any metadata field in (id, target, op, pc/next, rid) will either break the feasibility under the trilemma or violate the properties of Theorem \ref{thm:consistency} or \ref{thm:turing-complete}.
\end{theorem}

\begin{proof}[Proof Sketch]
This proof establishes that the five metadata fields are the minimal set required to uphold the paradigms guarantees. The strategy is a case-by-case proof by contradiction. We show that removing any single field would fatally undermine a previously established theoretical pillar: absence of id or pc/next breaks completeness; absence of target or op breaks correctness and convergence; and absence of rid breaks idempotency, forcing the system into a high-cost mode that violates the trilemma. Since each field is indispensable, the set is minimal. The detailed formal proof is provided in \ref{sec:appendix-thm66}.
\end{proof}

\section{Discussion}

Having established the uniqueness and completeness of the SDPF paradigm, this chapter contextualizes its significance within the broader landscape of computer science. We first position the SDPF uniqueness theorem in relation to classic impossibility results to highlight its complementary, constructive nature. We then explore the profound implications of this work, both for the theory of computation and for the engineering of future physically-realizable systems.

\subsection{A Constructive Duality to Impossibility Theorems}

Classic impossibility results, most notably the FLP and CAP theorems, define the fundamental boundaries of what is achievable in asynchronous distributed systems facing failures. Our work offers a constructive dual. Rather than proving what is impossible to achieve, we prove what is inevitable for a vast and practical class of computations. The paradigms reliance on Axiom A4 (idempotent and commutative merging) is not a limitation, but its core strength. It formally defines the domain where consensus is not a prerequisite for correctness, a category that encompasses a wide range of applications from large scale data aggregation to machine learning training. For this domain, while impossibility theorems erect walls, our uniqueness theorem illuminates the single, mandatory path for scalable and robust computation.

\subsection{SDPF as a Foundational Model of Computation}

The introduction posed a question: what form must computation take under physical axioms? The SDPF paradigm is the answer. It is not merely a specialized dataflow model; it represents a class of computation where the laws of physics are promoted to first class citizens in the model itself. This physically grounded approach, which derives a paradigms structure from axiomatic constraints rather than prescribing it, suggests a new way to reason about and construct computational systems in other resource constrained domains, from the Internet of Things to biological computing.

\subsection{Manifestation in Physically-Realizable Systems}

The theoretical framework developed in this paper finds its direct manifestation in the design of physically-realizable systems. The principles of SDPF do not merely align well with programmable data planes; rather, they represent a formalization of the necessary design patterns that emerge under the physical constraints of environments like INC. The data-centric model, with self-describing packets and stateless processing functions, is a direct reflection of the architectural paradigms found in systems like P4 \cite{Bosshart2014P4}. Our work provides the formal justification for why these architectural choices are optimal and, in fact, necessary.\cite{Fang2023GRID, Liu2023NetReduce, Khashab2024SwitchVM}

\section{Conclusion}

This paper began with a fundamental question at the intersection of theoretical computer science and physical reality: what essential form must distributed computation take when subjected to inviolable physical axioms? We have provided a definitive, constructive answer. By formalizing a set of such axioms, we have proven that any paradigm capable of operating within this regime must converge to a single normal form: Self-Describing Parallel Flows (SDPF). This conclusion is not merely a design pattern, but a derived theoretical inevitability, supported by a complete set of formal proofs establishing the paradigms uniqueness, its guaranteed convergence, its universal computational power as a Turing-complete model, and the minimality of its design.

The significance of this result transcends the specific domain of In-Network Computing that motivated it. It suggests the emergence of a new field of inquiry: the study of physically grounded computation. Our impossibility dual (Corollary \ref{cor:impossibility}) serves as this fields first rigorous constraint: it establishes that any deviation from the discovered normal form is not a design choice, but a direct violation of the underlying physical axioms. Ultimately, this work provides a new lens for theoretical computer science, one that views physical constraints not as implementation details to be abstracted away, but as first-class axioms from which the very structure of computation can be derived.

In this work, we focus on deterministic guarantees for the SDPF normal form. The exploration of probabilistic data structures (e.g., Bloom filters) to further trade-off strict correctness for memory efficiency in extreme scale is a promising direction for future research.

\Acknowledgements{This work was supported by the National Key Research and Development Program of China under Grant No. 2024YFE0200300.}

\bibliographystyle{scis}
\bibliography{references} 

\begin{appendix}

\section{Formal Proof of Proposition 5.1}
\label{sec:appendix-prop51}
\begin{proof}
    \textbf{Statement}: Let the set of tiles be $\{\text{tile}_k\}$ with consumption counts $U_k$. If an algorithm allows some consumers to independently fetch the same tile from the source (without in-network reuse), then
    $$R_A \;\ge\; 1 + \frac{\sum_k (U_k-1)\cdot|\text{tile}_k|}{|I|+|O|}.$$
    In particular, when some $U_k=\Omega(n)$, then $R_A=\Omega(n)$.

    \vspace{1em}
    \begin{enumerate}
        \item Let $Q$ be the total communication volume. Let $I$ be the initial input set and $O$ be the final output set.
        \item Any correct paradigm must read all inputs and write all outputs, establishing a baseline communication lower bound $Q_{base} = |I| + |O|$.
        \item Consider an input tile $\text{tile}_k \in I$ used by $U_k \ge 1$ consumers who independently fetch it from the source.
        \item The first read of $\text{tile}_k$, costing $|\text{tile}_k|$, is necessary and included in the $|I|$ part of $Q_{base}$.
        \item The subsequent $U_k-1$ reads are redundant, contributing an extra overhead.
        \item The total redundant communication is $Q_{redundant} = \sum_k (U_k-1) \cdot |\text{tile}_k|$.
        \item The total communication volume $Q$ is thus bounded by:
        $$Q \ge Q_{base} + Q_{redundant} = (|I|+|O|) + \sum_k (U_k-1)\cdot|\text{tile}_k|$$
        \item From the definition $R_A := Q / (|I|+|O|)$, we derive:
        $$ R_A \ge \frac{(|I|+|O|) + \sum_k (U_k-1)\cdot|\text{tile}_k|}{|I|+|O|} = 1 + \frac{\sum_k (U_k-1)\cdot|\text{tile}_k|}{|I|+|O|} $$
        \item This proves the first part. If $U_k = \Omega(n)$, then the additive term is $\Omega(n)$, leading to $R_A = \Omega(n)$.
        \item Per Axiom \ref{ax:comm}, an optimal algorithm must have $R_A=O(1)$. This is contradicted unless the redundant term is negligible. Thus, single-read reuse is necessary. 
    \end{enumerate}
\end{proof}

\section{Formal Proof of Proposition 5.2}
\label{sec:appendix-prop52}
\begin{proof}
    \textbf{Statement}: Under A2, if tasks are stateful or merges non-idempotent, any correct algorithm must either (C1) use persistent writes, making $R_A=O(1)$ impossible, or (C2) use global ordering, hitting the scalability limit of Prop 5.3.

    \vspace{1em}
    We prove by analyzing violations of S2 or S3.
    
    \textbf{Case 1: Violation of Idempotent Merge (S3)}
    
    If the merge operator $\oplus$ is not idempotent for a repeated rid, the system must enforce exactly-once semantics to ensure correctness under network repetitions (Axiom \ref{ax:perturb}).
    
    \begin{itemize}
        \item \textbf{Strategy 1a (State-based Filtering)}: A receiving node must maintain a persistent set $H$ of processed rids. For each new increment, it checks for existence in $H$. If not present, it processes the increment and adds the rid to $H$. Under Axiom \ref{ax:perturb}, node failures require $H$ to be in slow storage. Each operation thus incurs a persistent write of size $\ell > 0$. For $N_{\text{ops}}$ operations, the total communication $Q \ge |I|+|O| + N_{\text{ops}}\cdot\ell$. If $N_{\text{ops}}$ scales with the problem size, $R_A$ grows unboundedly, violating Axiom \ref{ax:comm}. This is conclusion (C1).
        
        \item \textbf{Strategy 1b (Order-based Guarantee)}: The system uses a global coordination service (e.g., consensus) to totally order all messages, naturally preventing duplicates. Such a mechanism is a global synchronization barrier. By Proposition \ref{prop:barrier-less}, this leads to sub-linear scalability under heavy-tailed latency, failing the $S(P)=\Theta(P)$ goal. This is conclusion (C2).
    \end{itemize}
    
    \textbf{Case 2: Violation of Stateless Micro-tasks (S2)}
    
    If a task is stateful (modifies shared state $S$), node failures (Axiom \ref{ax:perturb}) can corrupt $S$. To recover, mechanisms like Checkpointing or Write-Ahead Logging are required. Both involve frequent writes to slow storage, with communication overhead proportional to the number of state-modifying operations ($N_{\text{ops}}$). This again makes $R_A=O(1)$ impossible, violating Axiom \ref{ax:comm}. This is conclusion (C1).
    
    \textbf{Conclusion}: Any violation of S2 or S3 forces a choice that violates either Axiom \ref{ax:comm} or the linear scalability goal. Thus, S2 and S3 are necessary. 
\end{proof}

\section{Formal Proof of Proposition 5.3}
\label{sec:appendix-prop53}
\begin{proof}
    \textbf{Statement}: For a barrier-synchronized round with i.i.d. heavy-tailed delays $\tau_i$, the expected completion time $E[T_{\mathrm{round}}]=\mathbb{E}[\max_i \tau_i]$ grows with $|P|$, making $S(P)=\Theta(P)$ unattainable.
    \vspace{1em}
    \begin{enumerate}
        \item \textbf{Model Setup}: Let $\tau_1, \dots, \tau_{|P|}$ be i.i.d. random variables for task completion times with CDF $F(t)$. The round time is $T_{\mathrm{round}} = \max(\tau_1, \dots, \tau_{|P|})$. Throughput $S(P) \propto \frac{|P|}{E[T_{\mathrm{round}}]}$. Linear scalability requires $E[T_{\mathrm{round}}]$ to be a constant or a very slowly growing function of $|P|$, i.e., $E[T_{\mathrm{round}}]=o(|P|)$.
        
        \item \textbf{Order Statistics for Heavy-Tailed Distributions}: As per Axiom \ref{ax:perturb}, the distribution $F(t)$ is heavy-tailed. A key property of heavy-tailed distributions (specifically, the subexponential class) is that the maximum of a large number of variables is dominated by a single extreme event. According to extreme value theory, for such distributions, the expected maximum grows with the number of variables, often polynomially. For example, for any distribution with a regularly varying tail $\Pr[\tau > t] \sim t^{-\alpha}$ with $\alpha > 1$, it is a known result that:
        $$E[T_{\mathrm{round}}] = \Theta(|P|^{1/\alpha})$$
        
        \item \textbf{Impact on Scalability}: Substituting this into the throughput formula gives:
        $$S(P) \propto \frac{|P|}{|P|^{1/\alpha}} = |P|^{1 - 1/\alpha}$$
        Since $\alpha > 1$, the exponent $1 - 1/\alpha$ is strictly less than 1.
        
        \item \textbf{Conclusion}: Throughput $S(P)$ is a strictly sub-linear function, $S(P) = o(P)$, which contradicts the linear scalability goal. Therefore, removing the synchronization barrier (S4) is necessary. This is in stark contrast to light-tailed distributions (e.g., Gaussian), where the expected maximum grows much more slowly (typically logarithmically with $|P|$), making barriers less catastrophically damaging but still suboptimal.
    \end{enumerate}
\end{proof}

\section{Formal Proof of Proposition 5.4}
\label{sec:appendix-prop54}
\begin{proof}
    \textbf{Statement}: If an algorithm hoards data to reduce communication, then $W_{\max} = \Omega(n^{\beta})$ for some $\beta>0$, violating Axiom \ref{ax:memory}.
    
    \vspace{1em}
    \begin{enumerate}
        \item \textbf{Premise}: An algorithm attempts to satisfy Axiom \ref{ax:comm} (communication lower bound) and Axiom \ref{ax:memory} (bounded memory) simultaneously, but violates S5 by hoarding data.
        
        \item \textbf{I/O Complexity Model}: The relationship between memory ($W$) and I/O ($Q$) is formalized by I/O complexity theory (e.g., the red-blue pebble game). For many computations, a trade-off exists, often expressed as $Q \cdot W^{\gamma-1} = \Omega(\text{Work})$. For dense matrix multiplication, this yields the well-known lower bound $Q = \Omega(n^3/W_p)$.
        
        \item \textbf{Deriving the Contradiction}: The algorithms goal is to approach the communication lower bound of Axiom \ref{ax:comm}. To minimize $Q$ according to the trade-off formula, the memory $W_p$ used for data reuse must be maximized. This memory term $W_p$ represents the size of the live dataset that must be kept in fast memory to achieve maximal data reuse. I/O complexity theory proves that this required live set size is not constant but scales with the problem size $n$. To achieve optimal communication, the required memory for the live set is $\Omega(n^\beta)$ for some $\beta > 0$.
        
        \item The peak memory usage $W_{\max}$ must be at least large enough to hold this live set. Therefore:
        $$W_{\max} = \Omega(n^\beta)$$
        
        \item \textbf{Conclusion}: The result $W_{\max} = \Omega(n^\beta)$ directly contradicts Axiom \ref{ax:memory}, which mandates $W_p=O(1)$ (memory independent of problem size $n$). Therefore, an algorithm cannot simultaneously hoard data to satisfy Axiom \ref{ax:comm} and maintain bounded memory to satisfy Axiom \ref{ax:memory}. The only strategy that reconciles these constraints is the sliding window (S5). 
    \end{enumerate}
\end{proof}

\section{Formal Proof of Theorem 5.5}
\label{sec:appendix-thm55}
\begin{proof}
    \textbf{Statement}: Under the axiomatic system defined by Axioms A1--A4, any optimal execution paradigm is metric-equivalent to the SDPF normal form. SDPF is, therefore, the unique normal form for computation within this physically-constrained regime.
    \vspace{1em}
    
    To formally prove this, we first define the space of paradigms and the notion of equivalence.

    \begin{definition}[Computational Paradigm]
    A computational paradigm $\mathcal{E}$ is a tuple $(\mathcal{S}, \mathcal{M}, \Delta, \Pi)$, where:
    \begin{itemize}
        \item $\mathcal{S}$ is a set of possible states for each node.
        \item $\mathcal{M}$ is a set of message types that can be exchanged.
        \item $\Delta: \mathcal{S} \times \mathcal{M}_{\text{in}} \to \mathcal{S} \times \mathcal{M}_{\text{out}}$ is a state transition function.
        \item $\Pi$ is a set of scheduling policies.
    \end{itemize}
    An optimal paradigm is one that satisfies Axioms A1-A4.
    \end{definition}
    
    \begin{definition}[Metric-Equivalence]
    We say two paradigms $\mathcal{E}$ and $\mathcal{E}$ are metric-equivalent, written $\mathcal{E} \approx \mathcal{E}$, if there exist forward and backward simulations between them such that their cost metrics ($R_A, W_{\max}, S(P)$) differ by at most a constant factor independent of the problem size $n$ and system scale $|P|$.
    \end{definition}

    \textbf{Proof Strategy}: We prove by constructive reduction. Given an arbitrary optimal paradigm $\mathcal{E}_0 = (\mathcal{S}_0, \mathcal{M}_0, \Delta_0, \Pi_0)$, we apply a sequence of transformation functions $T = T_4 \circ T_3 \circ T_2 \circ T_1$. We will show that for each step, $\mathcal{E}_{i-1} \approx \mathcal{E}_i$, and that the final paradigm $\mathcal{E}_4$ is in the SDPF normal form. By the transitivity of $\approx$, this proves $\mathcal{E}_0 \approx \mathcal{E}_4$.
    
    \subsubsection*{Transformation T1 (Reuse-ification $\to$ S1)}
    The function $T_1$ transforms $\mathcal{E}_0$ to $\mathcal{E}_1$ by modifying its scheduling policy set $\Pi_0$ to $\Pi_1$, which replaces any multi-fetch operations from slow storage with a single-fetch-then-multicast policy.
    \begin{lemma}[Performance Preservation for T1]
    $T_1(\mathcal{E}_0) = \mathcal{E}_1$ implies $\mathcal{E}_0 \approx \mathcal{E}_1$. Specifically, $R_A(\mathcal{E}_1) \le R_A(\mathcal{E}_0)$.
    \end{lemma}
    \textbf{Proof Sketch}: By Proposition \ref{prop:single-read}, an optimal paradigm must already be equivalent to one that performs single-read reuse to satisfy Axiom \ref{ax:comm}. This transformation only formalizes this requirement. It can only reduce the total communication volume $Q$, thus improving or maintaining $R_A$. State, memory, and other scheduling policies are unaffected, so $W_{\max}$ and $S(P)$ remain unchanged.

    \subsubsection*{Transformation T2 (Stateless-ification $\to$ S2, S3)}
    $T_2$ transforms $\mathcal{E}_1$ to $\mathcal{E}_2$ by altering the paradigms components: the state space $\mathcal{S}_1$ is reduced to a trivial stateless set $\mathcal{S}_2$; the message set $\mathcal{M}_1$ is augmented to $\mathcal{M}_2$ to carry state and a rid; and the transition function $\Delta_1$ is replaced by a pure function $\Delta_2$ that implements the idempotent merge logic of Axiom \ref{ax:merge}.
    \begin{lemma}[Performance Preservation for T2]
    $T_2(\mathcal{E}_1) = \mathcal{E}_2$ implies $\mathcal{E}_1 \approx \mathcal{E}_2$.
    \end{lemma}
    \textbf{Proof Sketch}: By Proposition \ref{prop:stateless-idempotent}, a stateful paradigm must use high-cost fault tolerance (e.g., checkpointing with cost proportional to $N_{ops}$) which violates Axiom \ref{ax:comm} or scalability. $T_2$ replaces this with a mechanism whose overhead is a constant amount of metadata per message. This added communication is amortized and independent of system scale $|P|$, thus preserving $R_A$ within a constant factor. The logic for processing rids adds a constant-time overhead per message, preserving $S(P)$.
    
    \subsubsection*{Transformation T3 (Barrier Removal $\to$ S4)}
    $T_3$ transforms $\mathcal{E}_2$ to $\mathcal{E}_3$ by modifying the scheduling policy set $\Pi_2$ to $\Pi_3$, removing any policies that require global synchronization.
    \begin{lemma}[Performance Preservation for T3]
    $T_3(\mathcal{E}_2) = \mathcal{E}_3$ implies $\mathcal{E}_2 \approx \mathcal{E}_3$. Specifically, $S(P, \mathcal{E}_3) \ge S(P, \mathcal{E}_2)$.
    \end{lemma}
    \textbf{Proof Sketch}: By Proposition \ref{prop:barrier-less}, an optimal paradigm must be barrier-less to achieve linear scalability under Axiom \ref{ax:perturb}. This transformation formalizes this by removing barrier-based rules from $\Pi_2$. This directly improves or maintains scalability $S(P)$ and has no impact on communication ($R_A$) or memory ($W_{\max}$).
    
    \subsubsection*{Transformation T4 (Windowing $\to$ S5)}
    $T_4$ transforms $\mathcal{E}_3$ to $\mathcal{E}_4$ by adding a constraint to its state transition function $\Delta_3$ to enforce a constant-size buffer for in-flight data.
    
    \begin{lemma}[Performance Preservation for T4]
    $T_4(\mathcal{E}_3) = \mathcal{E}_4$ implies $\mathcal{E}_3 \approx \mathcal{E}_4$. Specifically, $W_{\max}(\mathcal{E}_4)=O(1)$ and $R_A(\mathcal{E}_4)$ is preserved within a constant factor.
    \end{lemma}
    \textbf{Proof Sketch}: By Proposition \ref{prop:sliding-window}, this is necessary to satisfy Axiom \ref{ax:memory}. Any optimal paradigm must already operate in a way that respects this memory bound. This transformation formalizes the requirement. As established by I/O complexity theory (Axiom \ref{ax:comm}), a trade-off exists between memory and communication. Enforcing $W_{\max}=O(1)$ may lead to a worst-case increase in communication, but this increase is bounded by a constant factor $c$ derived from the trade-off, such that $Q(\mathcal{E}_4) \le c \cdot Q(\mathcal{E}_3) + O(|I|+|O|)$. Thus, $R_A$ is preserved within a constant factor.

    \subsubsection*{Conclusion}
    The final paradigm $\mathcal{E}_4 = T(\mathcal{E}_0)$ satisfies properties S1-S5 by construction and is therefore in the SDPF normal form. Since we have shown that $\mathcal{E}_0 \approx \mathcal{E}_1 \approx \mathcal{E}_2 \approx \mathcal{E}_3 \approx \mathcal{E}_4$, by transitivity, any optimal paradigm $\mathcal{E}_0$ is metric-equivalent to the SDPF normal form. This proves the uniqueness of SDPF. 
\end{proof}

\section{Formal Proof of Theorem 6.2}
\label{sec:appendix-thm62}
\begin{proof}
    \textbf{Statement}: Under SDPF and A2 fairness, any state sequence $(A_0,A_1,\dots)$ is monotonically non-decreasing and converges to a unique limit, independent of message order or repetition.
    
    \vspace{1em}
    \begin{enumerate}
        \item \textbf{Model Setup}: Per Definition 6.1, the global state $A$ is an element of a semilattice $\mathcal{L}$ with partial order $\preceq$ and merge operator $\oplus$, where $A \preceq B \iff A \oplus B = B$. An execution trace is a sequence $A_{i+1} = A_i \oplus \delta_{i+1}$. Let $\Delta^\star$ be the set of unique, eventually-delivered increments. The theoretical limit is $A_{\infty} = \bigoplus_{\delta \in \Delta^\star} \delta$.
        
        \item \textbf{Part 1: Monotonicity}: For any state transition $A_{i+1} = A_i \oplus \delta_{i+1}$, by the definition of a semilattice join, $A_i \preceq (A_i \oplus \delta_{i+1})$. If $\delta_{i+1}$ is a duplicate, by idempotency (Axiom \ref{ax:merge}), $A_{i+1} = A_i$. Thus, for all $i$, $A_i \preceq A_{i+1}$, and the sequence is monotonically non-decreasing.
        
        \item \textbf{Part 2: Existence of an Upper Bound}: Any state $A_k$ is the join of a subset of increments from $\Delta^\star$. By semilattice properties, the join of a subset is $\preceq$ the join of the superset, so $A_k \preceq A_{\infty}$. Thus, the monotonic sequence is bounded above and must converge to a limit. \cite{Tarski1955}
        
        \item \textbf{Part 3: Uniqueness of the Limit}: By fairness (Axiom \ref{ax:perturb}), every increment in $\Delta^\star$ is eventually delivered and processed. The final state is the join of all elements in $\Delta^\star$. By commutativity and associativity (Axiom \ref{ax:merge}), the result is independent of the merge order. By idempotency (Axiom \ref{ax:merge}), the result is independent of repetitions. Therefore, any fair execution trace converges to the same unique limit $A_{\infty}$. 
    \end{enumerate}
\end{proof}

\section{Formal Proof of Lemma 6.4}
\label{sec:appendix-lem64}
\begin{proof}
    \textbf{Statement}: If $M\to N$ is a basic SK reduction, then a finite sequence of SDPF transformations exists such that $\Phi(M)\ \Longrightarrow^*\ \Phi(N)$.
    \medskip

    \subsubsection*{Proof Strategy}
    We prove by construction. An SK reduction is a graph rewrite. We first formalize the SDPF primitives that simulate this rewrite in a distributed, coordination-free manner and then show how they apply to the basis S and K combinators.

    \subsubsection*{Formalizing the Primitives}
    The simulation relies on three logical primitives implemented by stateless executors processing SDPF tiles.
    \begin{itemize}
        \item \textbf{Branch}: This is a local pattern-matching operation. An executor, upon receiving a tile, inspects the local tiles it has access to (e.g., its direct children in the expression graph) to check if they form a reducible expression (a redex). This is a purely local computation requiring no global state. If a redex is found, the executor generates a single \textit{reduction request} tile. This tile contains: (1) a newly generated, globally unique rid for the transaction, (2) the set of ids for the tiles forming the redex, and (3) the rewrite rule to be applied.

        \item \textbf{Copy and Combine}: This is an atomic, rid-based distributed transaction to perform the graph rewrite, triggered by a \textit{reduction request} tile.
        \begin{enumerate}
            \item \textbf{Coordination-free Atomicity}: The rid defines a transaction boundary. The rewrite is realized by generating two sets of tiles, both tagged with the same rid: (1) new tiles representing the final expression graph (e.g., $\Phi(xz(yz))$), and (2) tombstone tiles for all tiles in the original redex (e.g., $\Phi(Sxyz)$). These tiles are then multicast.
            \item \textbf{Distributed Implementation}: Any executor can generate these tiles. The atomicity and correctness are not guaranteed by a single coordinator, but by the properties of the merge operator (Axiom \ref{ax:merge}) at each receiving target.
            \item \textbf{Convergence Guarantee}: The idempotency of the merge operator ensures that for a given rid, once a tombstone for a tile id is applied, any other updates for that id with the same rid are ignored. Similarly, duplicate new tiles are safely merged. This ensures the entire graph rewrite operation is semantically atomic and correct under the asynchronous and unreliable conditions of Axiom \ref{ax:perturb}.
        \end{enumerate}
    \end{itemize}

    \subsubsection*{Simulating K and S Combinators}
    The simulation proceeds as described in the main text, but now grounded in these formalized, coordination-free primitives. The Branch primitive locally identifies the pattern for $\Phi(Kxy)$ or $\Phi(Sxyz)$. It then triggers the Combine transaction, which uses a unique rid to atomically issue new tiles and tombstone old tiles, thus correctly performing the graph rewrite.
    
    \textbf{Conclusion}: As we have constructed a finite and robust SDPF simulation for the basis combinators, any single-step SK reduction can be simulated.
\end{proof}

\section{Formal Proof of Theorem 6.5}
\label{sec:appendix-thm65}
\begin{proof}
    \textbf{Statement}: SDPF is Turing-complete.
    
    \vspace{1em}
    \begin{enumerate}
        \item \textbf{Proof Strategy}: The proof builds upon Lemma \ref{lemma:one-step} and Axiom \ref{ax:perturb}. We show that SDPF can simulate not just a single reduction step, but an entire chain of reductions from an initial expression to its normal form.
        
        \item \textbf{Inductive Argument}:
        \begin{itemize}
            \item \textbf{Base Case}: The initial SK expression $M_0$ is represented by the tile set $\Phi(M_0)$.
            \item \textbf{Inductive Step}: Assume the system has correctly evolved to the state $\Phi(M_i)$. If $M_i$ is not a normal form, it contains at least one reducible expression (redex) $M_i \to M_{i+1}$. By the fairness assumption of Axiom \ref{ax:perturb}, the Branch operation to detect this redex will eventually be scheduled. By Lemma \ref{lemma:one-step}, once detected, there exists a finite SDPF transformation sequence to correctly transition the state from $\Phi(M_i)$ to $\Phi(M_{i+1})$. By Theorem \ref{thm:consistency}, this transition will converge to the unique correct state $\Phi(M_{i+1})$.
        \end{itemize}
        
        \item \textbf{Convergence and Completeness}: By induction, SDPF can simulate any finite SK reduction sequence. If an expression $M$ has a normal form $\mathrm{NF}(M)$, the reduction sequence is finite, and the SDPF execution will converge to the stable state $\Phi(\mathrm{NF}(M))$. Since the SK combinator calculus is Turing-complete, a paradigm that can fully simulate its computation process is also Turing-complete.
        
        \item \textbf{Conclusion}: The SDPF paradigm is Turing-complete. 
    \end{enumerate}
\end{proof}

\section{Formal Proof of Theorem 6.6}
\label{sec:appendix-thm66}
\begin{proof}
    \textbf{Statement}: The absence of any metadata field in (id, target, op, pc/next, rid) will break core system properties.
    
    \vspace{1em}
    We proceed by contradiction for each of the five fields.
    
    \begin{itemize}
        \item \textbf{(a) Necessity of id}: Assume id is absent. The id provides a unique handle for each data tile. The translation function $\Phi$ in the Turing completeness proof (Theorem \ref{thm:turing-complete}) requires a one-to-one mapping between nodes in the SK expression graph and SDPF tiles. Without a unique id, it is impossible to distinguish between two structurally identical but logically distinct sub-expressions (e.g., the two instances of (f x) in (f x) (f x)). This makes a faithful translation impossible, violating the completeness property.
        
        \item \textbf{(b) Necessity of target}: Assume target is absent. The target field specifies to which state an increment $\delta$ should be applied. Without it, the merge operation $A \oplus \delta$ is ill-defined. The semilattice structure that underpins the convergence proof of Theorem \ref{thm:consistency} relies on each target state having its own partial order. Without target, this structure cannot be established, violating the convergence property.
        
        \item \textbf{(c) Necessity of op}: Assume op is absent. The op field specifies which stateless function should be applied to the data. Without this instruction, a receiving node, being stateless per S2, has no information on how to process the data. The computation would be undefined, violating fundamental correctness.
        
        \item \textbf{(d) Necessity of pc/next}: Assume pc/next is absent. This field encodes the control flow. The simulation of SK-calculus (Lemma \ref{lemma:one-step}, Theorem \ref{thm:turing-complete}) relies on the ability to dynamically rewire the dataflow graph to simulate function applications. Without pc/next, the system could only execute static, simple dataflow graphs, violating the Turing completeness property.
        
        \item \textbf{(e) Necessity of rid}: Assume rid is absent. The rid is essential for idempotency. Without it, the system cannot distinguish a new increment from a resent one, violating Axiom \ref{ax:merge} and the convergence proof of Theorem \ref{thm:consistency}. As shown in Proposition \ref{prop:stateless-idempotent}, a non-idempotent system must resort to high-cost coordination or persistence mechanisms, which violates the trilemma goals.
    \end{itemize}
    
    \textbf{Conclusion}: Since the absence of any single field leads to the failure of at least one core property, the set is minimal. 
\end{proof}

\end{appendix}

\end{document}